\documentclass[12pt,onecolumn,draftcls]{IEEEtran}

\usepackage{amssymb}
\usepackage{amsmath,epsfig}

\begin{document}

\newtheorem{lemma}{Lemma}
\newtheorem{theorem}{Theorem}
\newtheorem{corollary}[theorem]{Corollary}
\newtheorem{Thm}{\bf Theorem}
\newcommand{\bo}[1]{\boldsymbol{#1}}
\newcommand{\beq}{\begin{equation}}
\newcommand{\eeq}{\end{equation}}
\newcommand{\etal}{{\it  et al.\ }}
\newcommand{\lb}{\left(}
\newcommand{\rb}{\right)}
\newcommand{\lsb}{\left[}
\newcommand{\rsb}{\right]}
\newcommand{\la}{\left\{ }
\newcommand{\ra}{\right\} }
\newcommand{\lan}{\left\langle }
\newcommand{\ran}{\right\rangle }
\newcommand{\lbg}{\left\lceil}
\newcommand{\rbg}{\right\rceil}
\newcommand{\defi}{\stackrel{\bigtriangleup}{=}}
\newcommand{\Range}{{Range\,}}
\newcommand{\diag}{\,\mbox{diag}\,}
\newcommand{\blockdiag}{\,\mbox{blockdiag}\,}
\newcommand{\diagb}{\,\overline{\mbox{diag}}\,}
\newcommand{\E}{\,\mbox{E}\;}
\newcommand{\MSE}{\mbox{MSE}}
\newcommand{\TS}{N_{TS}}

\newcommand{\dsum}{\displaystyle\sum}
\newcommand{\difrac}{\displaystyle\frac}
\newcommand{\dmin}{\displaystyle\min}
\newcommand{\dmax}{\displaystyle\max}
\newcommand{\dlim}{\displaystyle\lim}
\newcommand{\dinf}{\displaystyle\inf}

\def\adots{\mathinner{\mskip0mu\raise0pt\vbox{\kern7pt\hbox{.}}\mskip3mu
          \raise4pt\hbox{.}\mskip3mu\raise8pt\hbox{.}\mskip0mu}}
\newcommand{\m}{{\!\, -\!\,}}
\newcommand{\p}{{\!\, +\!\,}}
\newcommand{\T}{H}
\newcommand{\tT}{T}
\newcommand{\tcc}{*}
\newcommand{\w}{w}
\newcommand{\z}{{\it z}}        
\newcommand{\eps}{\epsilon}
\newcommand{\oalpha}{\boldsymbol{\alpha}}
\newcommand{\at}{@@}
\newcommand{\pac}{\dagger}      
\newcommand{\Lu}{\underline{L}}
\newcommand{\rank}{\mbox{rank}}
\newcommand{\bfPb}{\overline{\bfP}}

\newcommand{\lla}{l_1}
\newcommand{\llb}{l_2}
\newcommand{\llc}{l_3}
\newcommand{\lld}{l_4}
\newcommand{\lle}{l_5}
\newcommand{\llf}{l_6}
\newcommand{\llg}{l_7}
\newcommand{\llh}{l_8}
\newcommand{\lli}{l_9}
\newcommand{\llj}{l_10}

\newcommand{\bmf}[1]{\mbox{\boldmath ${#1}$}}
\newcommand{\bmfscript}[1]{\mbox{\scriptsize{\boldmath $\displaystyle{#1}$}}}
\newcommand{\mbf}[1]{\mbox{\boldmath ${#1}$}}
\newcommand{\mbfscript}[1]{\mbox{\scriptsize{\boldmath $\displaystyle{#1}$}}}
\newcommand{\bdmA}{\mbox{\boldmath $A$}}
\newcommand{\bdmB}{\mbox{\boldmath $B$}}
\newcommand{\bdmC}{\mbox{\boldmath $C$}}
\newcommand{\bdmD}{\mbox{\boldmath $D$}}
\newcommand{\bdmE}{\mbox{\boldmath $E$}}
\newcommand{\bdmF}{\mbox{\boldmath $F$}}
\newcommand{\bdmG}{\mbox{\boldmath $G$}}
\newcommand{\bdmH}{\mbox{\boldmath $H$}}
\newcommand{\bdmI}{\mbox{\boldmath $I$}}
\newcommand{\bdmJ}{\mbox{\boldmath $J$}}
\newcommand{\bdmK}{\mbox{\boldmath $K$}}
\newcommand{\bdmL}{\mbox{\boldmath $L$}}
\newcommand{\bdmM}{\mbox{\boldmath $M$}}
\newcommand{\bdmN}{\mbox{\boldmath $N$}}
\newcommand{\bdmO}{\mbox{\boldmath $O$}}
\newcommand{\bdmP}{\mbox{\boldmath $P$}}
\newcommand{\bdmQ}{\mbox{\boldmath $Q$}}
\newcommand{\bdmR}{\mbox{\boldmath $R$}}
\newcommand{\bdmS}{\mbox{\boldmath $S$}}
\newcommand{\bdmT}{\mbox{\boldmath $T$}}
\newcommand{\bdmU}{\mbox{\boldmath $U$}}
\newcommand{\bdmV}{\mbox{\boldmath $V$}}
\newcommand{\bdmW}{\mbox{\boldmath $W$}}
\newcommand{\bdmX}{\mbox{\boldmath $X$}}
\newcommand{\bdmY}{\mbox{\boldmath $Y$}}
\newcommand{\bdmZ}{\mbox{\boldmath $Z$}}


\newcommand{\bdma}{\mbox{\boldmath $a$}}
\newcommand{\bdmb}{\mbox{\boldmath $b$}}
\newcommand{\bdmc}{\mbox{\boldmath $c$}}
\newcommand{\bdmd}{\mbox{\boldmath $d$}}
\newcommand{\bdme}{\mbox{\boldmath $e$}}
\newcommand{\bdmf}{\mbox{\boldmath $f$}}
\newcommand{\bdmg}{\mbox{\boldmath $g$}}
\newcommand{\bdmh}{\mbox{\boldmath $h$}}
\newcommand{\bdmi}{\mbox{\boldmath $i$}}
\newcommand{\bdmj}{\mbox{\boldmath $j$}}
\newcommand{\bdmk}{\mbox{\boldmath $k$}}
\newcommand{\bdml}{\mbox{\boldmath $l$}}
\newcommand{\bdmm}{\mbox{\boldmath $m$}}
\newcommand{\bdmn}{\mbox{\boldmath $n$}}
\newcommand{\bdmo}{\mbox{\boldmath $o$}}
\newcommand{\bdmp}{\mbox{\boldmath $p$}}
\newcommand{\bdmq}{\mbox{\boldmath $q$}}
\newcommand{\bdmr}{\mbox{\boldmath $r$}}
\newcommand{\bdms}{\mbox{\boldmath $s$}}
\newcommand{\bdmt}{\mbox{\boldmath $t$}}
\newcommand{\bdmu}{\mbox{\boldmath $u$}}
\newcommand{\bdmv}{\mbox{\boldmath $v$}}
\newcommand{\bdmw}{\mbox{\boldmath $w$}}
\newcommand{\bdmx}{\mbox{\boldmath $x$}}
\newcommand{\bdmy}{\mbox{\boldmath $y$}}
\newcommand{\bdmz}{\mbox{\boldmath $z$}}


\newcommand{\srfA}{\mbox{$\mathsf A$}}
\newcommand{\srfB}{\mbox{$\mathsf B$}}
\newcommand{\srfC}{\mbox{$\mathsf C$}}
\newcommand{\srfD}{\mbox{$\mathsf D$}}
\newcommand{\srfE}{\mbox{$\mathsf E$}}
\newcommand{\srfF}{\mbox{$\mathsf F$}}
\newcommand{\srfG}{\mbox{$\mathsf G$}}
\newcommand{\srfH}{\mbox{$\mathsf H$}}
\newcommand{\srfI}{\mbox{$\mathsf I$}}
\newcommand{\srfJ}{\mbox{$\mathsf J$}}
\newcommand{\srfK}{\mbox{$\mathsf K$}}
\newcommand{\srfL}{\mbox{$\mathsf L$}}
\newcommand{\srfM}{\mbox{$\mathsf M$}}
\newcommand{\srfN}{\mbox{$\mathsf N$}}
\newcommand{\srfO}{\mbox{$\mathsf O$}}
\newcommand{\srfP}{\mbox{$\mathsf P$}}
\newcommand{\srfQ}{\mbox{$\mathsf Q$}}
\newcommand{\srfR}{\mbox{$\mathsf R$}}
\newcommand{\srfS}{\mbox{$\mathsf S$}}
\newcommand{\srfT}{\mbox{$\mathsf T$}}
\newcommand{\srfU}{\mbox{$\mathsf U$}}
\newcommand{\srfV}{\mbox{$\mathsf V$}}
\newcommand{\srfW}{\mbox{$\mathsf W$}}
\newcommand{\srfX}{\mbox{$\mathsf X$}}
\newcommand{\srfY}{\mbox{$\mathsf Y$}}
\newcommand{\srfZ}{\mbox{$\mathsf Z$}}


\newcommand{\srfa}{\mbox{$\mathsf a$}}
\newcommand{\srfb}{\mbox{$\mathsf b$}}
\newcommand{\srfc}{\mbox{$\mathsf c$}}
\newcommand{\srfd}{\mbox{$\mathsf d$}}
\newcommand{\srfe}{\mbox{$\mathsf e$}}
\newcommand{\srff}{\mbox{$\mathsf f$}}
\newcommand{\srfg}{\mbox{$\mathsf g$}}
\newcommand{\srfh}{\mbox{$\mathsf h$}}
\newcommand{\srfi}{\mbox{$\mathsf i$}}
\newcommand{\srfj}{\mbox{$\mathsf j$}}
\newcommand{\srfk}{\mbox{$\mathsf k$}}
\newcommand{\srfl}{\mbox{$\mathsf l$}}
\newcommand{\srfm}{\mbox{$\mathsf m$}}
\newcommand{\srfn}{\mbox{$\mathsf n$}}
\newcommand{\srfo}{\mbox{$\mathsf o$}}
\newcommand{\srfp}{\mbox{$\mathsf p$}}
\newcommand{\srfq}{\mbox{$\mathsf q$}}
\newcommand{\srfr}{\mbox{$\mathsf r$}}
\newcommand{\srfs}{\mbox{$\mathsf s$}}
\newcommand{\srft}{\mbox{$\mathsf t$}}
\newcommand{\srfu}{\mbox{$\mathsf u$}}
\newcommand{\srfv}{\mbox{$\mathsf v$}}
\newcommand{\srfw}{\mbox{$\mathsf w$}}
\newcommand{\srfx}{\mbox{$\mathsf x$}}
\newcommand{\srfy}{\mbox{$\mathsf y$}}
\newcommand{\srfz}{\mbox{$\mathsf z$}}

\newcommand{\oneb}{\mbox{\boldmath $1$}}
\newcommand{\zerb}{\mbox{\boldmath $0$}}

\newcommand{\Us}{\mbox{$\mathcal U$}}
\newcommand{\Vs}{\mbox{$\mathcal V$}}


\newcommand{\bdSA}{\mbox{\bf A}}
\newcommand{\bdSB}{\mbox{\bf B}}
\newcommand{\bdSC}{\mbox{\bf C}}
\newcommand{\bdSD}{\mbox{\bf D}}
\newcommand{\bdSE}{\mbox{\bf E}}
\newcommand{\bdSF}{\mbox{\bf F}}
\newcommand{\bdSG}{\mbox{\bf G}}
\newcommand{\bdSH}{\mbox{\bf H}}
\newcommand{\bdSI}{\mbox{\bf I}}
\newcommand{\bdSJ}{\mbox{\bf J}}
\newcommand{\bdSK}{\mbox{\bf K}}
\newcommand{\bdSL}{\mbox{\bf L}}
\newcommand{\bdSM}{\mbox{\bf M}}
\newcommand{\bdSN}{\mbox{\bf N}}
\newcommand{\bdSO}{\mbox{\bf O}}
\newcommand{\bdSP}{\mbox{\bf P}}
\newcommand{\bdSQ}{\mbox{\bf Q}}
\newcommand{\bdSR}{\mbox{\bf R}}
\newcommand{\bdSS}{\mbox{\bf S}}
\newcommand{\bdST}{\mbox{\bf T}}
\newcommand{\bdSU}{\mbox{\bf U}}
\newcommand{\bdSV}{\mbox{\bf V}}
\newcommand{\bdSW}{\mbox{\bf W}}
\newcommand{\bdSX}{\mbox{\bf X}}
\newcommand{\bdSY}{\mbox{\bf Y}}
\newcommand{\bdSZ}{\mbox{\bf Z}}

\newcommand{\bdSa}{\mbox{\bf a}}
\newcommand{\bdSb}{\mbox{\bf b}}
\newcommand{\bdSc}{\mbox{\bf c}}
\newcommand{\bdSd}{\mbox{\bf d}}
\newcommand{\bdSe}{\mbox{\bf e}}
\newcommand{\bdSf}{\mbox{\bf f}}
\newcommand{\bdSg}{\mbox{\bf g}}
\newcommand{\bdSh}{\mbox{\bf h}}
\newcommand{\bdSi}{\mbox{\bf i}}
\newcommand{\bdSj}{\mbox{\bf j}}
\newcommand{\bdSk}{\mbox{\bf k}}
\newcommand{\bdSl}{\mbox{\bf l}}
\newcommand{\bdSm}{\mbox{\bf m}}
\newcommand{\bdSn}{\mbox{\bf n}}
\newcommand{\bdSo}{\mbox{\bf o}}
\newcommand{\bdSp}{\mbox{\bf p}}
\newcommand{\bdSq}{\mbox{\bf q}}
\newcommand{\bdSr}{\mbox{\bf r}}
\newcommand{\bdSs}{\mbox{\bf s}}
\newcommand{\bdSt}{\mbox{\bf t}}
\newcommand{\bdSu}{\mbox{\bf u}}
\newcommand{\bdSv}{\mbox{\bf v}}
\newcommand{\bdSw}{\mbox{\bf w}}
\newcommand{\bdSx}{\mbox{\bf x}}
\newcommand{\bdSy}{\mbox{\bf y}}
\newcommand{\bdSz}{\mbox{\bf z}}

\newcommand{\cA}{{\mathcal A}}
\newcommand{\cB}{{\mathcal B}}
\newcommand{\cC}{{\mathcal C}}
\newcommand{\cD}{{\mathcal D}}
\newcommand{\cE}{{\mathcal E}}
\newcommand{\cF}{{\mathcal F}}
\newcommand{\cG}{{\mathcal G}}
\newcommand{\cH}{{\mathcal H}}
\newcommand{\cI}{{\mathcal I}}
\newcommand{\cJ}{{\mathcal J}}
\newcommand{\cK}{{\mathcal K}}
\newcommand{\cL}{{\mathcal L}}
\newcommand{\cM}{{\mathcal M}}
\newcommand{\cN}{{\mathcal N}}
\newcommand{\cO}{{\mathcal O}}
\newcommand{\cP}{{\mathcal P}}
\newcommand{\cQ}{{\mathcal Q}}
\newcommand{\cR}{{\mathcal R}}
\newcommand{\cS}{{\mathcal S}}
\newcommand{\cT}{{\mathcal T}}
\newcommand{\cU}{{\mathcal U}}
\newcommand{\cV}{{\mathcal V}}
\newcommand{\cW}{{\mathcal W}}
\newcommand{\cX}{{\mathcal X}}
\newcommand{\cY}{{\mathcal Y}}
\newcommand{\cZ}{{\mathcal Z}}

\newcommand{\overbdmA}{\mbox{\boldmath $\bar{A}$}}
\newcommand{\Tau}{\mbox{$\cT$}}
\newcommand{\overtau}{\mbox{$\bar{\cT}$}}
\newcommand{\tiltau}{\mbox{$\widetilde{\cT}$}}
\newcommand{\overbdmV}{\mbox{\boldmath $\bar{V}$}}

\newcommand{\overdelta}{\mbox{$\bar{\delta}$}}

\newcommand{\sv}{\mbox{$\sigma^2_v$}}
\newcommand{\sa}{\mbox{$\sigma^2_a$}}
\newcommand{\sbb}{\mbox{$\sigma^2_b$}}
\newcommand{\sigs}{\sigma^2}
\newcommand{\sx}{\mbox{$\sigma^2_x$}}

\newcommand{\ah}{\widehat{a}}
\newcommand{\bh}{\widehat{b}}
\newcommand{\Rh}{\widehat{R}}
\newcommand{\Yu}{\underline{Y}}
\newcommand{\Yb}{\overline{Y}}
\newcommand{\Su}{\underline{S}}
\newcommand{\Sba}{\overline{S}}
\newcommand{\bdmcu}{\underline{\bdmc}}
\newcommand{\bdmcb}{\overline{\bdmc}}
\newcommand{\oalphab}{\overline{\oalpha}}

\newcommand{\bdmfh}{\widehat{\bdmf}}

\newcommand{\bdmFt}{\widetilde{\bdmF}}
\newcommand{\bdmGt}{\widetilde{\bdmG}}
\newcommand{\bdmAt}{\widetilde{\bdmA}}

\newcommand{\bfH}{\mbox{\bf H}}
\newcommand{\bfF}{\mbox{\bf F}}
\newcommand{\bfG}{\mbox{\bf G}}
\newcommand{\bfY}{\mbox{\bf Y}}
\newcommand{\bfP}{\mbox{\bf P}}
\newcommand{\bfQ}{\mbox{\bf Q}}
\newcommand{\bfR}{\mbox{\bf R}}
\newcommand{\bfU}{\mbox{\bf U}}
\newcommand{\bfV}{\mbox{\bf V}}
\newcommand{\bfA}{\mbox{\bf A}}
\newcommand{\bfB}{\mbox{\bf B}}
\newcommand{\bfW}{\mbox{\bf W}}
\newcommand{\bfT}{\mbox{\bf T}}
\newcommand{\bfD}{\mbox{\bf D}}
\newcommand{\bfC}{\mbox{\bf C}}
\newcommand{\bfL}{\mbox{\bf L}}
\newcommand{\bfN}{\mbox{\bf N}}
\newcommand{\bfI}{\mbox{\bf I}}
\newcommand{\bfE}{\mbox{\bf E}}
\newcommand{\bfZ}{\mbox{\bf Z}}
\newcommand{\bfX}{\mbox{\bf X}}
\newcommand{\bfM}{\mbox{\bf M}}
\newcommand{\bfO}{\mbox{\bf O}}
\newcommand{\bfJ}{\mbox{\bf J}}

\newcommand{\bfSigma}{\mbox{\boldmath $\Sigma$}}
\newcommand{\bfLambda}{\mbox{\boldmath $\Lambda$}}
\newcommand{\bfPhi}{\mbox{\boldmath $\Phi$}}
\newcommand{\bfalpha}{\mbox{\boldmath $\alpha$}}
\newcommand{\bfDN}{\mbox{\bf DN}}

\newfont{\bb}{msbm10 scaled 1100}
\newcommand{\RR}{\mbox{\bb R}}
\newcommand{\CC}{\mbox{\bb C}}
\newcommand{\FF}{\mbox{\bb F}}

\newcommand{\bfy}{\mbox{\bf y}}
\newcommand{\bfh}{\mbox{\bf h}}
\newcommand{\bff}{\mbox{\bf f}}
\newcommand{\bfv}{\mbox{\bf v}}
\newcommand{\bfx}{\mbox{\bf x}}
\newcommand{\bfa}{\mbox{\bf a}}
\newcommand{\bfb}{\mbox{\bf b}}
\newcommand{\bfr}{\mbox{\bf r}}
\newcommand{\bfe}{\mbox{\bf e}}
\newcommand{\bfw}{\mbox{\bf w}}
\newcommand{\bfc}{\mbox{\bf c}}
\newcommand{\bfu}{\mbox{\bf u}}
\newcommand{\bfg}{\mbox{\bf g}}
\newcommand{\bfs}{\mbox{\bf s}}
\newcommand{\bft}{\mbox{\bf t}}
\newcommand{\bfz}{\mbox{\bf z}}

\newcommand{\bmP}{\mbox{\boldmath $P$}}

\newcommand{\doplus}{\displaystyle\bigoplus}
\newcommand{\Nh}{\widehat{N}}
\newcommand{\bfhh}{\widehat{\bfh}}
\newcommand{\hh}{\widehat{\bfh}}
\newcommand{\bfhhh}{\widehat{\bfhh}}
\newcommand{\htt}{\widetilde{\bfh}}
\newcommand{\oXX}{\overline{\mbox{\bf XX}}}
\newcommand{\oCov}{\overline{Cov}}

\newcommand{\dint}{\displaystyle\int}
\newcommand{\doint}{\displaystyle\oint}
\newcommand{\dprod}{\displaystyle\prod}
\newcommand{\cBb}{\overline{\cB}}
\newcommand{\Cb}{\overline{C}}
\newcommand{\bfS}{\mbox{\bf S}}
\newcommand{\bfHh}{\widehat{\bfH}}
\newcommand{\bfHhh}{\widehat{\bfHh}}
\newcommand{\bfhl}{\overline{\bfh}}
\newcommand{\Hh}{\widehat{\bfH}}
\newcommand{\Ht}{\widetilde{\bfH}}
\newcommand{\Hb}{\overline{\bfH}}
\newcommand{\bfht}{{\widetilde{\bfh}\rule{0mm}{3.4mm}}}
\newcommand{\Vt}{\widetilde{\bfV}}
\newcommand{\bfYb}{\overline{\bfY}}
\newcommand{\bfXb}{\overline{\bfX}}

\newcommand{\bfyh}{\widehat{\bfy}}
\newcommand{\bfyt}{\widetilde{\bfy}}
\newcommand{\Sh}{\widehat{S}}
\newcommand{\St}{\widetilde{S}}
\newcommand{\bfRh}{\widehat{\bfR}}
\newcommand{\bfRt}{\widetilde{\bfR}}
\newcommand{\bfrt}{\overline{\bfr}}

\newcommand{\tr}{\mbox{tr}}
\newcommand{\vect}{\mbox{vec}}
\newcommand{\SNR}{\mbox{SNR}}
\newcommand{\MFB}{\mbox{MFB}}

\newcommand{\xh}{\widehat{x}}
\newcommand{\bfbh}{\widehat{\bfb}}
\newcommand{\bfbt}{\widetilde{\bfb}}
\newcommand{\bt}{\widetilde{b}}
\newcommand{\bfLb}{\overline{\bfL}}
\newcommand{\dbfb}{\Delta{\bfb}}
\newcommand{\dbfc}{\Delta{\bfc}}

\newcommand{\Ns}{N_{s}}
\newcommand{\Nt}{N_{t}}
\newcommand{\Nr}{N_{r}}

\newcommand{\bmh}{\mbox{\boldmath $h$}}

\newcommand{\xb}{\overline{x}}
\newcommand{\yb}{\overline{y}}

\newcommand{\CN}{{\cal CN}}
\newcommand{\SINR}{\mbox{SINR}}
\newcommand{\prob}{\mbox{Prob}}
\newcommand{\bfHb}{\overline{\bfH}}
\newcommand{\bfHt}{\widetilde{\bfH}}
\newcommand{\Lt}{\widetilde{L}}
\newcommand{\Dt}{\widetilde{D}}
\newcommand{\dt}{\widetilde{d}}

\title{Average Minimum Transmit Power to achieve SINR Targets: Performance Comparison of Various User Selection Algorithms}
\author{Umer~Salim,~\IEEEmembership{Member,~IEEE,}
        and~Dirk~Slock,~\IEEEmembership{Fellow,~IEEE}
\thanks{
Umer Salim is currently working at Infineon Technologies, 2600 Route des Cr\^etes, 06560 Sophia Antipolis, France (email: umer.salim@infineon.com). Dirk Slock is with Mobile Communications Department of EURECOM, France. (email: dirk.slock@eurecom.fr). This research was mostly conducted when Umer was a doctoral student at Eurecom, France. A limited part of the material in this paper appears in \cite{u_Asilomar09} and was presented at the Asilomar Conference on Signals, Systems, and Computers, 2009.
}
}
\maketitle
\begin{abstract}
In multi-user communication from one base station (BS) to multiple users, the problem of minimizing the transmit power to achieve some target guaranteed performance (rates) at users has been well investigated in the literature. Similarly various user selection algorithms have been proposed and analyzed when the BS has to transmit to a subset of the users in the system, mostly for the objective of the sum rate maximization.

We study the joint problem of minimizing the transmit power at the BS to achieve specific signal-to-interference-and-noise ratio (SINR) targets at users in conjunction with user scheduling. The general analytical results for the average transmit power required to meet guaranteed performance at the users' side are difficult to obtain even without user selection due to joint optimization required over beamforming vectors and power allocation scalars. We study the transmit power minimization problem with various user selection algorithms, namely semi-orthogonal user selection (SUS), norm-based user selection (NUS) and angle-based user selection (AUS). When the SINR targets to achieve are relatively large, the average minimum transmit power expressions are derived for NUS and SUS for any number of users. For the special case when only two users are selected, similar expressions are further derived for AUS and a performance upper bound which serves to benchmark the performance of other selection schemes. Simulation results performed under various settings indicate that SUS is by far the better user selection criterion.
\end{abstract}        
\section{Introduction}
\label{sec:intro}
\subsection{Motivation}
In multi-antenna downlink (DL) systems, the characterization of the capacity (rate) regions and the maximization of the sum rate have been among the most widely studied subjects. The capacity region of DL single antenna systems was first studied by Cover in \cite{cover_72}. After the discovery of spatial multiple antenna gains for single-user (SU) systems in \cite{telatar} \cite{Foschini}, the focus of research shifted to multiple antenna multi-user (MU) systems. Conditioned upon the availability of perfect channel state information (CSI), the capacity region of multi-antenna DL channel is known \cite{weingarten_bc} \cite{tse_bc} \cite{caire_bc} \cite{cioffi_bc} and hence the optimal (dirty paper coding (DPC), first proposed in \cite{costa_DPC} was shown to be the optimal strategy in \cite{weingarten_bc}) and a wide variety of sub-optimal (but less complicated) transmission strategies have been proposed and analyzed. In many practical wireless systems, maximizing the throughput may not be the primary objective. A very important design objective for multi-antenna MU systems is to achieve a particular link quality over all links with minimum transmission power which is equivalent to achieving certain signal-to-interference-and-noise ratios (SINR) or data rates over corresponding links. This problem, in some sense, is the dual problem of the sum rate maximization under a fixed power constraint. Certainly from an operator's perspective, the minimization of average transmit power to achieve these SINR targets is of prime importance.

Combined MU transmission with user scheduling has been widely analyzed in the sum rate maximization perspective (see \cite{yoo_ZF_BF} \cite{yoo_ZF_Scheduling} and the references therein) but very rarely for the objective of the transmit power minimization. Very pertinent questions in this area include how does the minimum average transmit power decay with the number of users or the number of BS transmit antennas. Similarly the optimal user selection scheme for transmit power minimization has never been investigated. In the context of the sum rate maximization, the semi-orthogonal user selection (SUS) has been shown to behave asymptotically optimal \cite{yoo_ZF_BF} and is widely believed to be the best greedy user selection strategy but no such study has been conducted for the transmit power optimization problem with hard SINR targets and no analytical results for average transmit power are known. Hence the characterization of the average minimum transmit power for various user selection mechanisms and relative performance comparisons are very relevant research objectives.
\subsection{The State of the Art}
The problem of minimizing the DL transmit power required to meet users' SINR constraints by joint optimization of transmit beamforming (BF) vectors and power allocation scalars was solved in \cite{schubert_TVT04} and \cite{schubert}. They showed the interesting duality of uplink (UL) and DL channels for this problem. Exploiting this UL-DL duality, they gave iterative algorithms to find the optimal BF matrices and the optimal power assignments to the users and showed the convergence of these algorithms to the optimal solution. For Gaussian MU channels (either UL or DL), they showed that the problem of minimizing the transmit power to achieve specific SINR targets bears a relatively simple solution due to the added structure which may be exploited by successive interference cancellation (SIC) in the UL and by DPC based encoding for known interference in the DL channels and the results were presented in \cite{schubert_VTC02}, \cite{schubert_ISSSTA02} and \cite{schubert}. The optimal BF strategy turns out to be the minimum-mean-square-error (MMSE) solution where each user will see no interference from the already encoded users (DPC based encoding) and each BF treats the interference of unencoded users as extra noise, and power allocation for each user is done to raise its SINR level to the target SINR. Actually the DL problem is solved by first solving the dual UL problem due to its relatively simple structure.

The performance of different user selection algorithms for transmit power minimization was studied in \cite{zhang_PIMRC07} (\cite{zhang_tx_power_eurasip} is the journal version). The Gaussian MU systems were analyzed without exploiting the extra system structure through SIC or DPC when SINR targets are large. They obtained analytical expressions for the average minimum transmit power required for guaranteed rates with norm-based user selection (NUS) and angle-based user selection (AUS) in the limiting case when only two users are selected. For the same scenario of two selected users, the expressions for average minimum transmit power were derived for NUS, AUS and SUS employing SIC (in UL) or DPC (in DL) in \cite{u_Asilomar09}. 
\subsection{Contribution}
We study the problem of average transmit power minimization to meet users' SINR constraints in conjunction with user scheduling. In this Gaussian MU system, we make use of SIC in the UL channel or DPC based encoding in the DL channel. As the channel information is already required at the BS for BF and power allocation assignments, this additional processing does not require any extra information. This problem formulation gives twofold advantage over \cite{zhang_PIMRC07}: first no iterations are required to compute the optimal BF vectors and power allocation scalars, and second less average power is required at the transmitter to satisfy the same SINR constraints at the users' side. Working under the similar setting of large SINR targets, the average minimum transmit power expressions are derived for any number of users selected through SUS, NUS or random user selection (RUS). These general results and a lemma about the instantaneous transmit power to achieve hard SINR targets become the main contributions of this work. For the case of two users transmitted simultaneously, we derive similar analytical expression with AUS. A performance upper bound is also derived for the two user case which may serve to benchmark any user selection mechanism. We compare the performance of these user selection algorithms in terms of average minimum transmit power required to satisfy users' SINR constraints. It turns out that NUS and AUS are strictly sub-optimal when compared with SUS.
\subsection{Organization}
This contribution is organized as follows. Section \ref{sec:model} describes the system model. Section \ref{sec:rev_power_min} gives a brief overview of the problem of transmit power minimization without user selection. In section \ref{sec:rev_selection}, certain user selection algorithms are reviewed for which later we analyze the performance. The main results of the chapter, the analytical expressions for the  average minimum transmit power for different user selection schemes, are presented in section \ref{sec:main}. The proof details have been relegated to appendices to keep the subject material simple and clear. The performances of these user selection algorithms are compared in section \ref{sec:perf} followed by the concluding remarks in section \ref{sec:conc}.
\section {System Model}
\label{sec:model}
The system, we consider, consists of a BS having $M$ transmit antennas and $K$ single-antenna user terminals. In the DL, the signal received by $k$-th user can be expressed as
\beq
y_k = \mathbf{h_k^{\dagger} x} + z_k, \hspace{1 cm} k = 1,2,\ldots, K
\label{eq:dl_sys}
\eeq
where $\mathbf{h_1^{\dagger}}$, $\mathbf{h_2^\dagger},\ldots,\mathbf{h_K^\dagger}$ are the channel vectors of users $1$ through user $K$ with $\mathbf{h_k} \in \mathbb{C}^{M \times 1}$, $\mathbf{x} \in \mathbb{C}^{M \times 1}$ denotes the signal transmitted by the BS and $z_1, z_2, \ldots, z_K$ are independent complex Gaussian additive noise terms with zero mean and variance $\sigma^2$. We denote the concatenation of the channels by $\mathbf{H_F^\dagger} = [\mathbf{h_1} \mathbf{h_2} \cdots \mathbf{h_K}]$, so $\mathbf{H_F}$ is  $K \times M$ forward channel matrix with $k$-th row equal to the channel of $k$-th user ($\mathbf{h_k^\dagger}$). The channel is assumed to be block fading having coherence length of $T$ symbol intervals. The entries of the forward channel matrix $\mathbf{H_F}$ are i.i.d. complex Gaussian with zero mean and unit variance. We make the simplifying assumption of the presence of perfect CSI at the transmitter (CSIT) so as to focus completely on the performance of different user selection algorithms.

The SINR constraints of the users are denoted by $\gamma_1$, $\gamma_2, \ldots \gamma_K$. As SINR is a direct measure of the successful signal decoding capability at a receiver (user), these constraints can be easily translated to rate constraints. If $K_s$ out of $K$ users (implying $K_s < K$) are selected for transmission during each coherence interval, the channel input $\mathbf{x}$ can be written as $\mathbf{x = \overline{V} P^{1/2} u}$, where $\mathbf{\overline{V}} \in \mathbb{C}^{M \times K_s}$ denotes the beamforming matrix with normalized columns, $\mathbf{P}$ is $K_s \times K_s$ diagonal power allocation matrix with positive real entries and $\mathbf{u} \in \mathbb{C}^{K_s \times 1}$ is the vector of zero-mean unit-variance Gaussian information symbols. Hence, $\mathbb{E} [\mathrm{Tr}(\mathbf{P})]$ is the average transmit power which can be minimized by optimizing over the beamforming matrix $\mathbf{\overline{V}}$ and the power allocation matrix $\mathbf{P}$ to achieve the SINR targets. We select this minimum average transmit power as the performance metric and study the performance of various user selection algorithms when users' SINR targets need to be satisfied.
\section{Overview of Transmit Power Minimization Problem}
\label{sec:rev_power_min}
The signal received by $k$-th user can be written as
\begin{eqnarray}
y_k & = & \mathbf{h_k^{\dagger} \overline{V} P^{1/2} u} + z_k, \hspace{1 cm} k = 1,2,\ldots, K_s \nonumber \\
    & = & \sqrt{p_k} \mathbf{h_k^{\dagger} \bar{v}_k} u_k + \sum_{\stackrel{j=1}{j\neq k}}^{K_s}\sqrt{p_j} \mathbf{h_k^{\dagger} \bar{v}_j} u_j + z_k, 
\label{eq:dl_sys2}
\end{eqnarray}
where $p_k$ represents the power allocated to the stream of $k$-th user. The second term in the expression represents the interference contribution at $k$-th user due to beams meant for other selected users. If the successive encoding at the transmitter is done from $K_s$ to $1$, then at $k$-th user it will receive the interference of those users which are encoded after this one. Hence the effective signal will be
\beq
y_k =  \sqrt{p_k} \mathbf{h_k^{\dagger} \bar{v}_k} u_k + \sum_{j=1}^{k-1}\sqrt{p_j} \mathbf{h_k^{\dagger} \bar{v}_j} u_j + z_k. 
\label{eq:dl_sys3}
\eeq
Based upon this received signal, the SINR of $k$-th user can be written as
\beq
\mathrm{SINR_k} = \frac{p_k |\mathbf{h_k^{\dagger} \bar{v}_k}|^2} {\displaystyle\sum_{j=1}^{k-1} p_j |\mathbf{h_k^{\dagger} \bar{v}_j}|^2 + \sigma^2 }.
\label{eq:sinr_k}
\eeq
Implicit in this SINR expression is the fact that the users are equipped with simple receivers which do not try to decode the signal of other users and hence the interference present in the received signal is treated as noise. Such receivers are commonly known in the literature as SU receivers \cite{tse_book}, \cite{cover_IT}. Without user selection, the problem of optimization of beamforming vectors and power allocation matrix was solved in \cite{schubert} and \cite{schubert_TVT04} using the UL-DL duality (see Section $4.3$ and $5.2$ in \cite{schubert} for details). They gave iterative algorithms to obtain the optimal beamforming vectors and the optimal power allocation for each user. The optimal beamforming vectors corresponding to a particular (sub-optimal) power allocation are obtained, then power allocations are updated corresponding to these beamforming vectors. This process is repeated till both converge to their optimal values. Unfortunately general closed form expressions for the transmit power required to achieve SINR targets don't exist due to intricate inter-dependence of beamforming vectors and power allocations, as is evident from eq. (\ref{eq:sinr_k}).

For Gaussian MU systems (the case of interest), the extra structure allows the use of SIC in UL or DPC based encoding in the DL. This permits to obtain the optimal BF vectors and power assignments using back substitution without any iteration. Although iterations are not required in this scenario, yet beamforming vector and power allocation of one user depend upon the BF vectors and power assignments of already treated users. If the noise variance at each user is $\sigma^2$, the minimum instantaneous transmit power required is given by the following expression taken from Section $5.2$ of \cite{schubert}.
\beq
p_{\mathrm{tx}}(\mathbf{h_1, h_2, \ldots h_{K_s}})  = \sigma^2 \sum_{i=1}^{K_s} \frac{\gamma_i}{\mathbf{h_i^{\dagger} Z_i^{-1} h_i}}, 
\eeq
where $\mathbf{Z_i}$ makes a subspace gathering the contributions from the channels of those users which will produce interference for $i$-th user and is given by the following expression
\beq
\mathbf{Z_i = I_M} + \sum_{j=1}^{i-1} p_j \mathbf{h_j h_j^{\dagger}}.
\eeq
The following lemma gives a sufficiently accurate approximation of the above given instantaneous power when SINR targets are relatively large.
\begin{lemma}[Minimum Instantaneous Transmit Power to achieve SINR Targets]
If the users' SINR targets are sufficiently large, the minimum instantaneous transmit power to achieve these targets for $K_s$ users can be closely approximated by the following expression:
\beq
p_{\mathrm{tx}}(\mathbf{h_1, h_2, \ldots h_{K_s}})  = \sigma^2 \sum_{i=1}^{K_s} \frac{\gamma_i}{\mathbf{||h_i||^2} \sin^2\theta_{(i-1)}}
\label{eq:inst_power}
\eeq
where $\theta_{(i-1)}$ is the angle which $\mathbf{h_i}$ subtends with the $(i-1)$-dimensional (interference) subspace spanned by $\mathbf{h_1, h_2 \ldots h_{i-1}}$ for $i > 1$ and $\theta_0 = \frac{\pi}{2}$. 
\end{lemma}
\begin{proof}
The proof details for this lemma appear in Appendix \ref{app:tx_power_lemma}.
\end{proof}
This lemma about the required transmit power to achieve SINR targets bears a very nice intuitive explanation. It says that the effective channel strength of each user (taking into account the interference streams that it has to deal with) is the energy in the projection of this user's channel when it is projected on the null space of its interference subspace, the subspace spanned by the channels of those users who create interference for this user (as a function of encoding order). In our setting where encoding order is $K_s$ to $1$, the interference subspace for user $i$ is the subspace spanned by the channels of users $1, 2 \ldots i-1$. Then each user is allocated the minimal power corresponding to its effective channel energy such that it achieves its SINR target. The sum of these powers gives the minimum instantaneous transmit power required to achieve SINR targets at $K_s$ active users.
%
\section{Review of User Selection Algorithms}
\label{sec:rev_selection}
There is a plethora of user selection algorithms in the literature and hence a comprehensive review is out of the scope of this paper. In this section, we briefly give the overview of three most famous user selection algorithms for which we later study the problem of transmit power minimization and derive the corresponding average power expressions.
\subsection{Norm-Based User Selection (NUS)}
\label{sec:rev_NUS}
In NUS, the users are selected based only upon their channel strengths. Hence $K$ users are sorted in the descending order of their channel norm values, and the first $K_s$ (strongest) users are selected for transmission in each scheduling interval.
\subsection{Angle-Based User Selection (AUS)}
\label{sec:rev_AUS}
The user selection criterion in AUS is the mutual orthogonality of users' channel vectors. The first user is selected which has the largest channel norm. The second user is selected as the one which is the most orthogonal to this user, without paying any regard to its channel strength. The third selected user is the one whose channel vector is the most orthogonal to the subspace spanned by the two already selected users' channels. This process is repeated till $K_s$ users have been selected.
\subsection{Semi-Orthogonal User Selection (SUS)}
\label{sec:rev_SUS}
The user selection metric for SUS is the combination of the channel strength and its spatial orthogonality with respect to the other users. The first selected user is the one with the largest channel norm. The second selected user would be the one whose projection on the null space of the first user has the largest norm. The third selected user will be the one whose projection on the null space of the subspace spanned by the channel vectors of the first two users has the largest norm. This process is repeated till $K_s$ users get selected. Interested readers can find the details of this algorithm in \cite{yoo_ZF_BF} or \cite{yoo_ZF_Scheduling}.
\subsection{Random User Selection (RUS)}
\label{sec:rev_RUS}
The RUS selects the active users independent of their channel realizations. Hence the active users can be selected following the round-robin algorithms for fairness in terms of being in the active pool or the active users can be selected based upon the users' subscription conditions (users paying more rates to service providers could be given some kind of priority over other users). Clearly this is not a good criteria of choosing the active users for the objective of minimum transmit power to achieve SINR targets but it will serve the purpose of performance lower bound for any user selection mechanism where users are selected based upon their channel conditions.
\section{Transmit Power with User Selection - Main Results}  
\label{sec:main}
In this section, we give the main results of this paper, the analytical expressions for the average minimum transmit power required to achieve specific SINR targets at users. The users are selected obeying different user selection algorithms as detailed in Section \ref{sec:rev_selection} and in the second step, we compute the optimal beamforming vectors and power assignments following the steps outlined in Section \ref{sec:rev_power_min}. We work with the assumption that all the users have the same SINR targets $\gamma$ otherwise the users with smaller SINR targets become relatively better candidates compared to those with higher targets for the objective of transmit power minimization for a fixed number of users treated simultaneously. The proofs have been relegated to appendices for simplicity and lucidity. The results for NUS, SUS and RUS are fully general and hold for any number of active users whereas for AUS and performance upper bound, we could only derive the results when two active users are selected for simultaneous transmission. 
\begin{theorem}[Average Minimum Transmit Power for NUS]
Consider a DL system having a BS equipped with $M$ transmit antennas and $K$ single antenna users, each having an SINR constraint of $\gamma$, and $K_s$ active users are selected for simultaneous transmission from the pool of $K$ users in each coherence block. If the active users are chosen through NUS, the average minimum transmit power, denoted as $p_{\mathrm{N}}(K_s)$, is given by:
\beq
p_{\mathrm{N}}(K_s)  =  \sigma^2 \gamma \sum_{i=1}^{K_s} \left( \mathbb{E}_{F_{||\mathbf{h}||^2}(M,K_s+1-i,K;x)} \left[\frac{1}{x}\right] \mathbb{E}_{F_{\sin^2\theta_{(i-1)}}(M;x)}\left[\frac{1}{x}\right] \right)
\eeq
where $F_{||\mathbf{h}||^2}(M,r,K;x)$ denotes the cumulative distribution function (CDF) of $r$-th order statistic of squared norm among $K$ independent $M$-dimensional complex Gaussian vectors and $F_{\sin^2\theta_j}(M;x)$ denotes the CDF of $\sin^2\theta_j$ where $\theta_j$ is the angle that an $M$-dimensional vector subtends with an independent $j$-dimensional subspace (possibly spanned by $j$ independent $M$-dimensional vectors). All these distributions have been grouped together in Appendix \ref{app:CDFs}.
\label{th:NUS}
\end{theorem}
%
\begin{corollary}[NUS for $2$ Users] When $K_s = 2$ active users are selected through NUS in each coherence block, the average minimum transmit power to achieve SINR target $\gamma$ is given by:
\beq
p_{\mathrm{N}}(2) =  \gamma \sigma^2 \left( K \alpha_{\mathrm{M,K-1}} - (K-2-\frac{1}{M-2}) \alpha_{\mathrm{M,K}} \right).
\eeq
where $\alpha_{\mathrm{M,K}}$ is a constant solely governed by $M$ and $K$ and is defined to be
\beq
\alpha_{\mathrm{M,K}} \stackrel{\Delta}{=} \int_0^{\infty} K \frac{e^{-x} x^{M-2}}{\Gamma(M)} \left[G(M,x)\right]^{K-1} dx,
\label{eq:alpha}
\eeq
where $\Gamma(M)$ and $G(M,x)$ denote the Gamma function and the regularized Gamma function \cite{abramowitz} respectively.
\label{th:NUS_cor_2}
\end{corollary}
%
\begin{corollary}[NUS for $4$ Users]  When $K_s = 4$ active users are selected through NUS in each coherence block, the average minimum transmit power to achieve SINR target $\gamma$ is given by:
\beq
\begin{array}{l}
p_{\mathrm{N}}(4) = \gamma \sigma^2 \left[ \frac{M-1}{M-4} \alpha_{\mathrm{M,K}} + \left\{ K \alpha_{\mathrm{M,K-1}} - (K-1) \alpha_{\mathrm{M,K}} \right\} \frac{M-1}{M-3} + \right. \nonumber \\ 
\left\{ \frac{K(K-1)}{2} \alpha_{\mathrm{M,K-2}} - K(K-2) \alpha_{\mathrm{M,K-1}} + \frac{(K-1)(K-2)}{2} \alpha_{\mathrm{M,K}} \right\} \frac{M-1}{M-2} + \nonumber \\
\left. \left\{ \frac{K(K-1)(K-2)}{6} \alpha_{\mathrm{M,K-3}} - \frac{K(K-1)(K-3)}{2} \alpha_{\mathrm{M,K-2}} + \frac{K(K-2)(K-3)}{2} \alpha_{\mathrm{M,K-1}} - \frac{(K-1)(K-2)(K-3)}{6} \alpha_{\mathrm{M,K}} \right\} \right].
\end{array}
\eeq
\label{th:NUS_cor_4}
\end{corollary}
\begin{proof}
The proof details for Theorem \ref{th:NUS} and its associated corollaries \ref{th:NUS_cor_2} and \ref{th:NUS_cor_4} appear in Appendix \ref{app:NUS}.\\
\end{proof}

\begin{theorem}[Average Minimum Transmit Power for SUS]
For an $M$ transmit antenna BS and $K$ single antenna users, if $K_s$ active users are selected through SUS for simultaneous transmission each having an SINR constraint of $\gamma$, the average minimum transmit power, denoted by $p_{\mathrm{S}}(K_s)$, is given by:
\beq
p_{\mathrm{S}}(K_s)  =  \sigma^2 \gamma \sum_{i=1}^{K_s} \left( \mathbb{E}_{F_{||\mathbf{h}||^2}(M+1-i,i,K;x)} \left[\frac{1}{x}\right] \right).
\eeq
\label{th:SUS}
\end{theorem}
%
\begin{corollary}[SUS for $2$ Users] When $K_s = 2$ active users are selected through SUS in each coherence block, the average minimum transmit power to achieve SINR target $\gamma$ is given by:
\beq
p_{\mathrm{S}}(2) =  \gamma \sigma^2 \left( \alpha_{\mathrm{M,K}} + K \alpha_{\mathrm{M-1,K-1}} - (K-1) \alpha_{\mathrm{M-1,K}} \right).
\eeq
\label{th:SUS_cor_2}
\end{corollary}
%
\begin{corollary}[SUS for $4$ Users] When $K_s = 4$ active users are selected through SUS in each coherence block, the average minimum transmit power to achieve the SINR targets is given by:
\beq
\begin{array}{l}
p_{\mathrm{S}}(4) = \gamma \sigma^2 \left[ \alpha_{\mathrm{M,K}} + K \alpha_{\mathrm{M-1,K-1}} - (K-1) \alpha_{\mathrm{M-1,K}} + \frac{K(K-1)}{2} \alpha_{\mathrm{M-2,K-2}} - K(K-2) \alpha_{\mathrm{M-2,K-1}} \right.\nonumber \\
+ \left. \frac{(K-1)(K-2)}{2} \alpha_{\mathrm{M-2,K}} + \frac{K(K-1)(K-2)}{6} \alpha_{\mathrm{M-3,K-3}} - \frac{K(K-1)(K-3)}{2} \alpha_{\mathrm{M-3,K-2}} + \frac{K(K-2)(K-3)}{2} \alpha_{\mathrm{M-3,K-1}} \right.\nonumber \\
- \left. \frac{(K-1)(K-2)(K-3)}{6} \alpha_{\mathrm{M-3,K}}  \right],
\end{array}
\eeq
\label{th:SUS_cor_4}
\end{corollary}
\begin{proof}
The proof details for Theorem \ref{th:SUS} and corollaries \ref{th:SUS_cor_2} and \ref{th:SUS_cor_4} appear in Appendix \ref{app:SUS}.\\
\end{proof}

\begin{theorem}[Average Minimum Transmit Power for RUS]
For an $M$-antenna transmitter BS having $K$ single-antenna users in the pool, when $K_s$ active users are selected randomly for simultaneous transmission, the average minimum transmit power required, denoted by $p_{\mathrm{R}}(K_s)$, so that each of $K_s$ users achieves its SINR target $\gamma$  is given by:
\beq
p_{\mathrm{R}}(K_s)  =  \sigma^2 \gamma \left( \mathbb{E}_{F_{||\mathbf{h}||^2}(M;x)} \left[\frac{1}{x}\right]\right) \sum_{i=1}^{K_s} \left( \mathbb{E}_{F_{\sin^2\theta_i}(M;x)} \left[\frac{1}{x}\right] \right) = \gamma \sigma^2 \sum_{i=1}^{K_s} \frac{1}{M-i}
\label{eq:RUS}
\eeq
where $F_{||\mathbf{h}||^2}(M;x)$ denotes the CDF of the squared norm of an $M$-dimensional complex Gaussian vector which is $\chi^2$ having $2M$ degrees of freedom.
\label{th:RUS}
\end{theorem}
\begin{proof}
The proof outline is given in Appendix \ref{app:RUS}.
\end{proof}

\begin{theorem}[Average Minimum Transmit Power for AUS]
Consider a DL system having a BS equipped with $M$ transmit antennas and $K$ single antenna users, each having an SINR constraint of $\gamma$, and $K_s = 2$ users are selected for simultaneous transmission in each coherence block. If the user selection is done through AUS, the minimum average transmit power is given by:
\beq
p_{\mathrm{A}} (2) = \gamma \sigma^2 \left( \frac{1}{K-1} (\frac{K}{M-1} - \alpha_{\mathrm{M,K}} ) +  \frac{(M-1)(K-1)\alpha_{\mathrm{M,K}}}{(M-1)(K-1)-1} \right)
\label{eq:AUS}
\eeq
\label{th:AUS}
\end{theorem}
\begin{proof}
The proof sketch appears in Appendix \ref{app:AUS}.
\end{proof}

\begin{theorem}[Performance Benchmark for $2$ Selected Users]
For a system with an $M$-antenna BS and $K$ single antenna users, a lower bound on the average transmit power (performance benchmark), in case of $K_s = 2$ active users, required to achieve SINR targets is given by:
\beq
p_{\mathrm{L}} (2) = \gamma \sigma^2 \left( K \alpha_{\mathrm{M,K-1}} - (K-1) (1-\frac{M-1}{(M-1)(K-1)-1}) \alpha_{\mathrm{M,K}} \right).
\label{eq:power_lb}
\eeq
\label{th:LB}
\end{theorem}
\begin{proof}
The proof outline is given in Appendix \ref{app:LB}.
\end{proof}
\section{Performance Comparison} 
\label{sec:perf}
In this subsection, we compare the performance of user selection algorithms treated in previous sections when the metric of interest is the average minimum transmit power required to satisfy users' SINR constraints.
\subsection{The case of $K_s = 2$ Selected Users} 
\label{sec:perf_2_users}
\begin{figure}[htbp]
	\begin{center}
		\includegraphics[width=8cm]{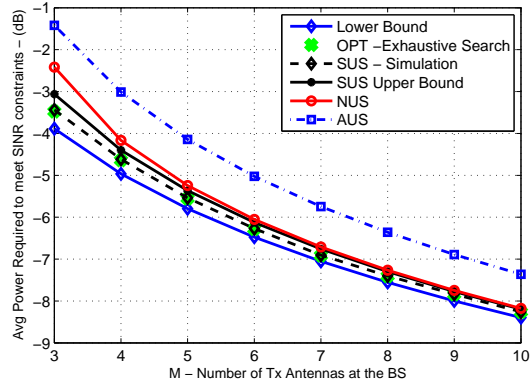}
	\end{center}
	\caption{Avg. Min. Transmit Power vs. M for $K=10$, $K_s=2$, $\gamma=10$ dB, $\sigma^2=0.1$. Curves show that SUS is the best strategy and follows closely the power lower bound. NUS also performs close to SUS with increasing number of transmit antennas.}
	\label{fig:P_vs_M_Ks=2}
\end{figure}
The plot of average minimum transmit power required to attain specific SINR targets $\gamma$ versus the number of antennas at the BS appears in Fig. \ref{fig:P_vs_M_Ks=2} for the considered user selection algorithms. We remark that SUS performs better than the other user selection schemes but with the increase in the number of transmit antennas, NUS also performs very well. The similar behaviour was observed in \cite{zhang_PIMRC07} and the reason comes from the fact that with the increase in the number of transmit antennas, users' channels start becoming (close to) spatially orthogonal (this is clearly visible through the angle distributions such as $F_{\sin^2\theta_i}(M;x)$ in appendix \ref{app:CDFs}) and furthermore, due to difference ($M-K_s$) very good beamforming vectors can be chosen to cause very small interference to other users.

\begin{figure}[htbp]
	\begin{center}
		\includegraphics[width=8cm]{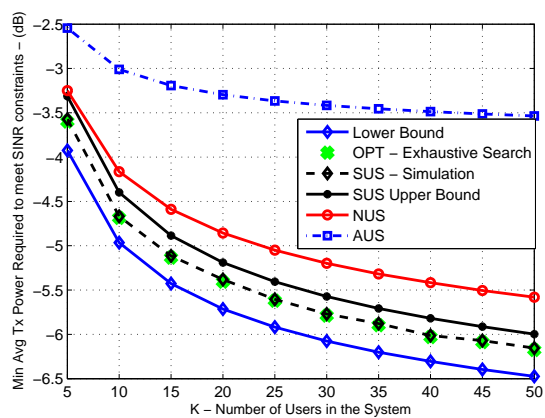}
	\end{center}
	\caption{Avg. Min. Transmit Power vs. Nb. of Users for $M=4$, $K_s=2$, $\gamma=10$ dB, $\sigma^2=0.1$. Curves show that SUS performs the best and NUS becomes sub-optimal when number of users increases.}
	\label{fig:P_vs_K_Ks=2}
\end{figure}
Fig. \ref{fig:P_vs_K_Ks=2} plots the curves of the minimum average transmit power versus the number of users for a fixed number of transmit antennas. SUS again performs very close to the optimal (obtained by exhaustive search) but we remark that NUS does not behave very well in this scenario because it just chooses users with good channel norms without paying any attention to their spatial orthogonality which may affect significantly the interference observed by the selected users.
\subsection{The Case of $K_s = 4$ Selected Users} 
\label{sec:perf_more_users}
\begin{figure}[htbp]
	\begin{center}
		\includegraphics[width=8cm]{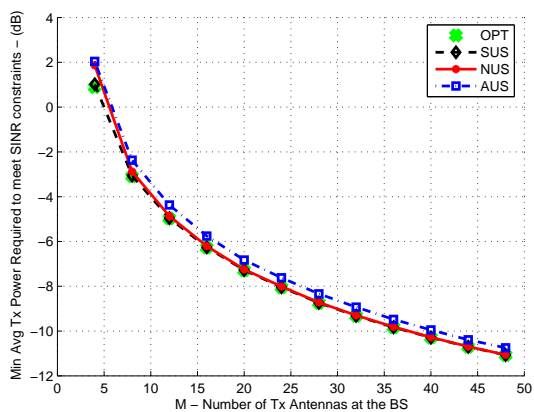}
	\end{center}
	\caption{Avg. Min. Transmit Power vs. M for $K=8$, $K_s=4$, $\gamma=10$ dB, $\sigma^2=0.1$. Curves show that SUS is the best strategy and follows closely the power lower bound. NUS also becomes optimal for a reasonably large number of transmit antennas.}
	\label{fig:P_vs_M_Ks=4}
\end{figure}

\begin{figure}[htbp]
	\begin{center}
		\includegraphics[width=8cm]{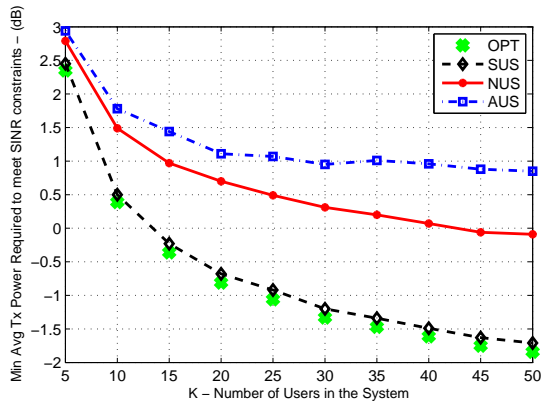}
	\end{center}
	\caption{Avg. Min. Transmit Power vs. Nb. of Users for $M=4$, $K_s=4$, $\gamma=10$ dB, $\sigma^2=0.1$. Curves show that SUS performs the best and NUS becomes strictly sub-optimal when the number of users increases.}
	\label{fig:P_vs_K_Ks=4}
\end{figure}

We plot the average minimum transmit power required to achieve certain SINR targets versus the number of transmit antennas and versus the number of system users in Fig. \ref{fig:P_vs_M_Ks=4} and Fig. \ref{fig:P_vs_K_Ks=4} respectively, for the user selection algorithms of interest. For both of these plots, the number of selected users is $4$. We observe the same behaviour as observed in the case of $2$ selected users. For large number of transmit antennas, both SUS and NUS perform very close to the optimal, even AUS achieves a reasonable performance. 

On the other hand, for a fixed number of transmit antennas at the BS when the number of users present in the system increases, the performance of NUS degrades substantially. The reason is that NUS captures the raw aspect of multi-user diversity which governs only the self signal power but pays no attention to the inter-user spatial separation which might have a larger impact on the interference power. The worst performance of AUS is expected as it pays no attention to the strength of the selected users which is quite important for power minimization objective. Moreover, SUS performs very close to the optimal, for any set of system parameters. The reason is the selection criterion of SUS where both the channel strength and the spatial orthogonality of the users are properly taken care of.
\section {Conclusions}
\label{sec:conc}
In this paper, we have studied the performance of various user selection algorithms in terms of the average minimum transmit power required to satisfy specific SINR targets at users' side. General closed form expressions of the average minimum transmit power for the three user selection algorithms, namely SUS, NUS and RUS, were derived when any number of users are selected for simultaneous transmission. Furthermore for the special case when only two users are selected for simultaneous transmission, similar expressions are derived for AUS and for performance upper bound which serves to benchmark other selection algorithms. SUS, which has been shown to behave close to optimal for the sum rate maximization objective under fixed power constraint, shows equally attractive performance in this dual problem setting of transmit power minimization to achieve hard SINR targets. For a fixed number of users and increasing number of transmit antennas, NUS performs very close to SUS. In the complementary setting of fixed number of BS transmit antennas and an increasing number of system users, NUS shows substantial performance degradation but SUS still performs very close to the optimal.
%
\appendices
\section{Proof of Minimum Instantaneous Transmit Power Lemma}
\label{app:tx_power_lemma}
The instantaneous transmit power required to achieve the SINR targets at $K_s$ active users having channels $\mathbf{h_1,h_2 \ldots h_{K_s}}$ is given by the following expression from \cite{schubert}
\beq
p_{\mathrm{tx}}(\mathbf{h_1, h_2, \ldots h_{K_s}})  = \sum_{i=1}^{K_s} \frac{\gamma_i}{\mathbf{h_i^{\dagger} Z_i^{-1} h_i}}, 
\eeq
where $\mathbf{Z_i}$ is given by the following expression
\beq
\mathbf{Z_i = I_M} + \sum_{j=1}^{i-1} p_j \mathbf{h_j h_j^{\dagger}}.
\eeq
We have further taken $\sigma^2 = 1$ following \cite{schubert} as it just appears as a constant scaling factor and may be absorbed in SINR targets as well.

The minimum power allocated to the stream of $1$st user to achieve its SINR target $\gamma_1$ is 
\beq
p_1 = \frac{\gamma_1}{\mathbf{h_1^{\dagger} I_M h_1}} = \frac{\gamma_1}{||\mathbf{h_1}||^2} .
\eeq
The power allocated to the stream of $2$nd user to achieve its SINR target $\gamma_2$ is 
\beq
p_2 = \frac{\gamma_2}{\mathbf{h_2^{\dagger} Z_2^{-1} h_2}} 
\eeq
Non-identity $\mathbf{Z_2^{-1}}$ appears because user $2$ will see the interference from the stream of $1$st user. 
\beq
\mathbf{Z_2^{-1}} = (\mathbf{I_M} + p_1 \mathbf{h_1 h_1^{\dagger})^{-1}}
\eeq
Applying the matrix inversion lemma (MIL) to the right hand side (R.H.S.) of the above equation, we get
\beq
\mathbf{Z_2^{-1}} = \mathbf{I_M} - p_1 \mathbf{h_1} (1+ p_1 ||\mathbf{h_1}||^2)^{-1} \mathbf{h_1^{\dagger}}.
\eeq
We can see from the power allocation to the stream of $1$st user that $\gamma_1 = p_1 ||\mathbf{h_1}||^2$. Thus according to the assumption made in the statement of this lemma if SINR target $\gamma_1$ is sufficiently large, the term $(1+ p_1 ||\mathbf{h_1}||^2)$ in the above equation can be closely approximated by $(p_1 ||\mathbf{h_1}||^2)$, though we use equality sign with some abuse of notation.
\beq
\mathbf{Z_2^{-1}} = \mathbf{I_M} - p_1 \mathbf{h_1} (p_1 ||\mathbf{h_1}||^2)^{-1} \mathbf{h_1^{\dagger}} = \mathbf{I_M} - \frac{\mathbf{h_1} \mathbf{h_1^{\dagger}}}{||\mathbf{h_1}||^2}
\eeq
This renders
\begin{eqnarray}
\mathbf{h_2^{\dagger} Z_2^{-1} h_2} & = & ||\mathbf{h_2}||^2 - |\mathbf{h_2^{\dagger}} \frac{\mathbf{h_1}}{||\mathbf{h_1}||}|^2 \nonumber \\
                                                             & = & ||\mathbf{h_2}||^2 (1 - \cos^2\theta_1) = ||\mathbf{h_2}||^2 \sin^2\theta_1,
\end{eqnarray}                                                             
where $\theta_1$ denotes the angle that $\mathbf{h_2}$ subtends with the $1$-dimensional subspace spanned by $\mathbf{h_1}$. Hence the power allocation done over the stream of $2$nd user so that it achieves its SINR target $\gamma_2$ would be
\beq
p_2 = \frac{\gamma_2}{||\mathbf{h_2}||^2 \sin^2\theta_1} 
\eeq
The power allocated to the stream of $3$rd user is given by 
\beq
p_3 = \frac{\gamma_3}{\mathbf{h_3^{\dagger}} \mathbf{Z_3^{-1} h_3}} 
\eeq
with
\beq
\mathbf{Z_3^{-1}} = (\mathbf{I_M} + p_1 \mathbf{h_1 h_1^{\dagger}}+ p_2 \mathbf{h_2 h_2^{\dagger}})^{-1}.
\eeq
Taking $\mathbf{\acute{h}_i} = \sqrt{p_i} \mathbf{h_i}$ and making a bigger matrix $\mathbf{H_{12} = [\acute{h}_1 \acute{h}_2]}$, we get $\mathbf{Z_3} = (\mathbf{I_M} + \mathbf{H_{12} H_{12}^{\dagger}})$. Applying MIL to the R.H.S. of the above equation, $\mathbf{Z_3^{-1}}$ can be written as
\beq
\mathbf{Z_3^{-1}} = \mathbf{I_M} - \mathbf{H_{12} (I_2+ H_{12}^{\dagger} H_{12})^{-1} H_{12}^{\dagger}}.
\eeq
$\mathbf{I_2}$ matrix adds $1$ to the diagonal elements of $\mathbf{H_{12}^{\dagger} H_{12}}$  which are $p_1 ||\mathbf{h_1}||^2$ and $p_2 ||\mathbf{h_2}||^2$ respectively. As $p_1 ||\mathbf{h_1}||^2 = \gamma_1$ and $p_2 ||\mathbf{h_2}||^2 = \frac{\gamma_2}{\sin^2\theta_1} > \gamma_2$, for large SINR targets the above equation will become
\beq
\mathbf{Z_3^{-1} = I_M - H_{12} (H_{12}^{\dagger} H_{12})^{-1} H_{12}^{\dagger}}.
\eeq
As $\mathbf{H_{12} (H_{12}^{\dagger} H_{12})^{-1} H_{12}^{\dagger}}$ is the projection matrix over the column space of $\mathbf{H_{12}}$ i.e., over the space spanned by $\mathbf{h_1}$ and $\mathbf{h_2}$, the product $\mathbf{h_3^{\dagger} Z_3^{-1} h_3}$ gives the energy of the channel $\mathbf{h_3}$ projected over the subspace orthogonal to that spanned by $\mathbf{h_1}$ and $\mathbf{h_2}$.
\begin{eqnarray}
\mathbf{h_3^{\dagger} Z_3^{-1} h_3} & = & ||\mathbf{h_3}||^2 - \mathbf{h_3^{\dagger} H_{12} (H_{12}^{\dagger} H_{12})^{-1} H_{12}^{\dagger} h_3} \nonumber \\
                                                             & = & ||\mathbf{h_3}||^2 (1 - \cos^2\theta_2) = ||\mathbf{h_3}||^2 \sin^2\theta_2
\end{eqnarray}                                                             
where $\theta_2$ is the angle subtended by $\mathbf{h_3}$ with the $2$-dimensional subspace spanned by $\mathbf{h_1}$ and $\mathbf{h_2}$. And hence the power allocated to the stream of $3$rd user to raise its SINR level to $\gamma_3$ is given by
\beq
p_3 = \frac{\gamma_3}{||\mathbf{h_3}||^2 \sin^2\theta_2}.
\eeq
In fact this procedure generalizes to any number of users and the power allocated to the stream of $i$-th active user is given by
\beq
p_i = \frac{\gamma_i}{||\mathbf{h_i}||^2 \sin^2\theta_{(i-1)}},
\eeq
where $\theta_{(i-1)}$ is the angle that $\mathbf{h_i}$ makes with the $(i-1)$-dimensional (interference) subspace spanned by $\mathbf{h_1, h_2} \ldots \mathbf{h_{i-1}}$ (this is a function of encoding order). Summing the powers allocated to all active $K_s$ users' streams, the total minimum instantaneous power to achieve SINR targets at $K_s$ users is given by:
\beq
p_{\mathrm{tx}}(\mathbf{h_1,h_2, \ldots h_{K_s}})  = \sum_{i=1}^{K_s} \frac{\gamma_i}{\mathbf{||h_i||^2} \sin^2\theta_{(i-1)}}.
\eeq
\section{Some Useful Distributions}
\label{app:CDFs}
In this appendix, we give some useful cumulative distribution functions (CDF) for which probability density functions (PDF) can be computed by simple differentiation. 
\subsection{Channel Norm Distributions}
Most of the channel norm (squared) distributions given in this subsection are known relations, others have been computed using the tools from order statistics \cite{H_David} and some of them also appear in \cite{zhang_PIMRC07}. If all the users have $M$-dimensional spatially i.i.d. complex Gaussian vector channels, the CDF of $||\mathbf{h_i}||^2$ for any $i$ is $\chi^2$ distributed with $2M$ degrees of freedom whose CDF is
\beq
F_{||\mathbf{h}||^2}(M;x)  = G(M,x),
\eeq
where G denotes the regularized Gamma function \cite{abramowitz}, and is defined as
\beq
G(M,x) = \frac{1}{\Gamma(M)} \int_0^x e^{-t} t^{M-1} dt.
\eeq
The PDF corresponding to CDF $F_{||\mathbf{h}||^2}(M;x)$ is given by
\beq
f_{||\mathbf{h}||^2}(M;x) = \frac{e^{-x} x^{M-1}}{\Gamma(M)}.
\eeq
Below we give the CDFs for the largest, the second largest, third and fourth order statistics. The CDF of the r-th largest order statistic among $K$ i.i.d. variables, each of which has the CDF of $F_{||\mathbf{h}||^2}(M;x)$, is given by \cite{H_David}
\beq
F_{||\mathbf{h}||^2}(M,r,K;x) = \sum_{j=K+1-r}^{K} \binom{K}{j} \left[F_{||\mathbf{h}||^2}(M;x)\right]^{j} \left[1-F_{||\mathbf{h}||^2}(M;x)\right]^{K-j}
\eeq

The CDF of the user having the largest channel norm among $K$ i.i.d. $M$-antenna users, denoted as $F_{||\mathbf{h}||^2}(M,1,K;x)$, each of whom is distributed as $F_{||\mathbf{h}||^2}(M;x)$ is given by
\beq
F_{||\mathbf{h}||^2}(M,1,K;x) = \left[F_{||\mathbf{h}||^2}(M;x)\right]^K.
\eeq
The CDF of the user having the second largest channel norm among $K$ i.i.d. users is 
\beq
F_{||\mathbf{h}||^2}(M,2,K;x) = K \left[F_{||\mathbf{h}||^2}(M;x)\right]^{K-1} - (K-1)\left[F_{||\mathbf{h}||^2}(M;x)\right]^K.
\eeq
The CDF of the user having the third largest channel norm among $K$ i.i.d. users is 
\beq
\begin{array}{l}
F_{||\mathbf{h}||^2}(M,3,K;x)  = \frac{K(K-1)}{2} \left[F_{||\mathbf{h}||^2}(M;x)\right]^{K-2} - K(K-2)\left[F_{||\mathbf{h}||^2}(M;x)\right]^{K-1} + \nonumber \\
\frac{(K-1)(K-2)}{2} \left[F_{||\mathbf{h}||^2}(M;x)\right]^{K} .
\end{array}
\eeq
The CDF of the user having the fourth largest channel norm among $K$ i.i.d. users is 
\beq
\begin{array}{l}
F_{||\mathbf{h}||^2}(M,4,K;x)  = \frac{K(K-1)(K-2)}{6} \left[F_{||\mathbf{h}||^2}(M;x)\right]^{K-3} - \frac{K(K-1)(K-3)}{2}\left[F_{||\mathbf{h}||^2}(M;x)\right]^{K-2} + \nonumber \\
\frac{K(K-2)(K-3)}{2} \left[F_{||\mathbf{h}||^2}(M;x)\right]^{K-1} - \frac{(K-1)(K-2)(K-3)}{6} \left[F_{||\mathbf{h}||^2}(M;x)\right]^{K-1} .
\end{array}
\eeq
The distribution of any random user among $K$ users which does not have the largest norm can be specified as (from \cite{zhang_PIMRC07})
\beq
F_{||\mathbf{h}||^2}(M,\acute{1},K;x)  = \frac{K}{K-1} F_{||\mathbf{h}||^2}(M;x) - \frac{1}{K-1} \left[F_{||\mathbf{h}||^2}(M;x)\right]^K
\label{eq:CDF_not_largest}
\eeq
where $\acute{1}$ stands for a random user which is not the first order statistic.
%
\subsection{Channel Direction Distributions}
In this subsection, we give some useful distributions of the $\sin^2$ and $\cos^2$ of the angle between a vector and a subspace. If we have $K$ i.i.d. $M$-dimensional Gaussian distributed vectors, i.e. $\mathbf{h_i} \in \mathbb{C}^M$ for user $i$, we can compute the distribution of the $\sin^2$ and $\cos^2$ of the angle between one vector and the subspace spanned by a subset of the other vectors. For a channel vector $\mathbf{h_j}$ and a subspace spanned by $i$ independent Gaussian vectors $\mathbf{h_1, h_2 \ldots, h_{i}}$, if $\theta_i$ denotes the angle $\mathbf{h_j}$ subtends with this $i$-dimensional subspace, the projection of $\mathbf{h_i}$ on this subspace $\cos^2\theta_i$ has a $\beta$ distribution with parameters $i$ and $M-i$ (see \cite{jindal_AntComb} and \cite{Gupta_beta} for details). $\sin^2\theta_i = 1- \cos^2\theta_i$ also has the beta distributions with shift of parameters $\beta\left(M-i,i\right)$. The distribution of $\sin^2\theta_i$, denoted as $F_{\sin^2\theta_i}(M;x)$ is given by
\beq
F_{\sin^2\theta_i}(M;x) = \frac{B_x(M-i,i)}{B(M-i,i)} = \frac{(M-1)!}{(M-i-1)!(i-1)!} \int_{t=0}^{x} t^{M-i-1} (1-t)^{i-1} dt,
\label{eq:angle_distr}
\eeq
where $B$ and $B_x$ denote the beta function and the regularized beta function respectively \cite{abramowitz} \cite{Gupta_beta}.

If $\theta_1$ denotes the angle that an $M$-dimensional vector $\mathbf{h_j}$ makes with an independent vector $\mathbf{h_1}$, the distribution of $\sin^2\theta_1$ is given by
\beq
F_{\sin^2\theta_1}(M;x) = x^{M-1}.
\eeq
If $\theta_2$ denotes the angle that $\mathbf{h_j}$ makes with a $2$-dimensional subspace spanned by two independent vectors $\mathbf{h_1}$ and $\mathbf{h_2}$, the distribution of $\sin^2\theta_2$ is given by
\beq
F_{\sin^2\theta_2}(M;x) = (M-1) x^{M-2} - (M-2) x^{M-1}.
\eeq
If $\theta_3$ denotes the angle that $\mathbf{h_j}$ makes with a $3$-dimensional subspace spanned by three independent vectors $\mathbf{h_1}$, $\mathbf{h_2}$ and $\mathbf{h_3}$, the distribution of $\sin^2\theta_3$ is given by
\beq
F_{\sin^2\theta_3}(M;x) = \frac{(M-1)(M-2)(M-3)}{2} \left( \frac{x^{M-1}}{M-1} - \frac{2 x^{M-2}}{M-2} + \frac{x^{M-3}}{M-3} \right).
\eeq
These distributions can be obtained by putting the appropriate value for the dimension of the subspace w.r.t. which orthogonalization is being performed in eq. (\ref{eq:angle_distr}).

We saw that the energy in the orthogonal projection of one vector over another independent vector assumes the CDF of $F_{\sin^2\theta_1}(M;x)$. If there are $K$ such projections (each with CDF of $F_{\sin^2\theta_1}(M;x)$), the CDF of the largest (1st order) projection is given by
\beq
F_{\sin^2\theta_1}(M,1,K;x) = \left[ F_{\sin^2\theta_1}(M;x) \right]^K = x^{K(M-1)}.
\eeq
\section{Norm-Based User Selection}
\label{app:NUS}
In the proof of the theorem for NUS and the rest of the appendices, we make extensive use of the useful CDFs which have been grouped together in appendix \ref{app:CDFs} so we highly encourage the readers to go through the previous appendix for proper understanding of these proofs and the notation associated to those CDFs.

For NUS, the users are chosen as described in section \ref{sec:rev_selection}. The squared norm of the first selected user is the largest among $K$ users and hence is distributed as $F_{||\mathbf{h}||^2}(M,1,K;x)$. Similarly the squared norms of the second, third and the fourth selected users are the 2nd, 3rd, 4th largest order statistics and hence distributed as $F_{||\mathbf{h}||^2}(M,2,K;x)$, $F_{||\mathbf{h}||^2}(M,3,K;x)$ and $F_{||\mathbf{h}||^2}(M,4,K;x)$ respectively. We reproduce the expression for minimum instantaneous transmit power below
\beq
p_{\mathrm{tx}}(\mathbf{h_1, h_2, \ldots h_{K_s}})  = \sigma^2 \gamma \sum_{i=1}^{K_s} \frac{1}{\mathbf{||h_i||^2} \sin^2\theta_{(i-1)}}.
\eeq
As these users are selected solely based upon their channel norms and the Gaussian distributed vectors have independent norms and directions, the directional properties of these vectors are as if they are randomly selected. Hence $\sin^2\theta_1$ is distributed as $F_{\sin^2\theta_1}(M;x)$, the distribution specified in Appendix \ref{app:CDFs}. Similarly $\sin^2\theta_2$ and $\sin^2\theta_3$ are distributed as the CDFs of $\sin^2\theta$ of a vector with a random $2$ and $3$-dimensional subspace, hence distributed as $F_{\sin^2\theta_2}(M;x)$ and $F_{\sin^2\theta_3}(M;x)$ respectively. 

For the case of multiple users, we have to perform SIC (considering UL) or DPC based encoding (considering DL) in a particular ordering. It's known that for the objective of the minimization of transmit power, the weaker user should be the one which gets decoded with the least interference \cite{tse_book}. This optimal ordering requires that the strongest user (distributed as $F_{||\mathbf{h}||^2}(M,1,K;x)$) should be the one facing the interference of all the users when its signal is decoded. This implies that its interference subspace would be ($K_s-1$)-dimensional and $\sin^2$ of the angle with this subspace would be distributed as $F_{\sin^2\theta_{(K_s-1)}}(M;x)$. Hence the average power corresponding to this strongest user (decoded with maximum interference) is given by
\beq
\mathbb{E} p_{\mathrm{tx}}(\mathbf{h_{K_s}})  = \sigma^2 \gamma   \mathbb{E}_{F_{||\mathbf{h}||^2}(M,1,K;x)} \left[\frac{1}{x}\right] \mathbb{E}_{F_{\sin^2\theta_{(K_s-1)}}(M;x)} \left[\frac{1}{x}\right].
\eeq

Similarly the weakest user has the distribution of $F_{||\mathbf{h}||^2}(M,K_s,K;x)$ and its signal gets decoded with no interference (as if it were alone). Hence the average transmit power allocated to the stream of this user is given by
\beq
\mathbb{E} p_{\mathrm{tx}}(\mathbf{h_{1}})  = \sigma^2 \gamma   \mathbb{E}_{F_{||\mathbf{h}||^2}(M,K_s,K;x)} \left[\frac{1}{x}\right] \mathbb{E}_{F_{\sin^2\theta_{(1-1)}}(M;x)} \left[\frac{1}{x}\right],
\eeq
where $\theta_0 = \frac{\pi}{2}$ by definition in the Lemma of transmit power.

For the user whose signal gets decoded with $(i-1)$ interference streams (decoded at $i$-th order) would be the one selected at $(K_s+1-i)$-th iteration of NUS, hence its squared norm would be distributed as $F_{||\mathbf{h}||^2}(M,K_s+1-i,K;x)$. As its interference subspace is $(i-1)$-dimensional, the $\sin^2$ of its angle with this subspace is distributed as $F_{\sin^2\theta_{(i-1)}}(M;x)$. This permits us to write the average transmit power allocated for this user to be
\beq
\mathbb{E} p_{\mathrm{tx}}(\mathbf{h}_{i})  = \sigma^2 \gamma   \mathbb{E}_{F_{||\mathbf{h}||^2}(M,K_s+1-i,K;x)} \left[\frac{1}{x}\right] \mathbb{E}_{F_{\sin^2\theta_{(i-1)}}(M;x)} \left[\frac{1}{x}\right].
\eeq

The average transmit powers allocated to the streams of all users can be summed up to get the total average minimum transmit power to achieve SINR targets when the users have been selected through NUS and is given by
\beq
p_{\mathrm{N}}(K_s)  =  \sigma^2 \gamma \sum_{i=1}^{K_s} \left( \mathbb{E}_{F_{||\mathbf{h}||^2}(M,K_s+1-i,K;x)} \left[\frac{1}{x}\right] \mathbb{E}_{F_{\sin^2\theta_{(i-1)}}(M;x)}\left[\frac{1}{x}\right] \right).
\eeq
\subsection{NUS for $2$ Users}
\label{app:NUS_2}
When only two users are selected through NUS for simultaneous transmission, the ordering strategy remains the same (the weaker user gets decoded with no interference) and the average power required can be computed by taking only the first two terms of the general NUS transmit power expression.
\beq
p_{\mathrm{N}} (2)  =  \sigma^2 \gamma \left( \mathbb{E}_{F_{||\mathbf{h}||^2}(M,2,K;x)} \left[\frac{1}{x}\right] + \mathbb{E}_{F_{||\mathbf{h}||^2}(M,1,K;x)} \left[\frac{1}{x}\right] \mathbb{E}_{F_{\sin^2\theta_{1}}(M;x)}\left[\frac{1}{x}\right] \right)
\label{eq:NUS_2_users}
\eeq
The PDF corresponding to CDF $F_{||\mathbf{h}||^2}(M,1,K;x)$ (obtained by its differentiation) is given by
\beq
f_{||\mathbf{h}||^2}(M,1,K;x) = K \left[G(M,x)\right]^{K-1} \frac{e^{-x} x^{M-1}}{\Gamma(M)}.
\eeq
It allows us to compute the following expectation:
\begin{eqnarray}
\mathbb{E}_{F_{||\mathbf{h}||^2}(M,1,K;x)} \left[\frac{1}{x}\right] & = & \int_0^{\infty} \left[\frac{1}{x}\right] K \left[G(M,x)\right]^{K-1} \frac{e^{-x} x^{M-1}}{\Gamma(M)} dx \nonumber \\
 & = & \int_0^{\infty} K \frac{e^{-x} x^{M-2}}{\Gamma(M)} \left[G(M,x)\right]^{K-1} dx  \nonumber \\
 & = & \alpha_{\mathrm{M,K}}
\end{eqnarray}
where the last equality is the definition of the constant term $\alpha_{\mathrm{M,K}}$, defined in eq. (\ref{eq:alpha}), which only depends upon the specific values of $M$ and $K$. Similarly it can be shown that 
\beq
\mathbb{E}_{F_{||\mathbf{h}||^2}(M,2,K;x)} \left[\frac{1}{x}\right] = K \alpha_{\mathrm{M,K-1}} - (K-1) \alpha_{\mathrm{M,K}}.
\eeq
The expectation concerning the angle distribution can also be computed as follows:
\beq
\mathbb{E}_{F_{\sin^2\theta_{1}}(M;x)}\left[\frac{1}{x}\right] = \int_0^1 \left[\frac{1}{x}\right] (M-1) x^{M-2} dx = \frac{M-1}{M-2}.
\eeq
Combining the results of these expectations in eq. (\ref{eq:NUS_2_users}) and doing some rearrangements gives the result of corollary \ref{th:NUS_cor_2}.
\subsection{NUS for $4$ Users}
\label{app:NUS_4}
The average transmit power when $4$ users are selected through NUS for simultaneous transmission can be computed by taking the first four terms from the general NUS average transmit power expression and computing the expectations.
\beq
\begin{array}{l}
p_{\mathrm{N}} (4)  =  \sigma^2 \gamma \left( \mathbb{E}_{F_{||\mathbf{h}||^2}(M,4,K;x)} \left[\frac{1}{x}\right] + \mathbb{E}_{F_{||\mathbf{h}||^2}(M,3,K;x)} \left[\frac{1}{x}\right] \mathbb{E}_{F_{\sin^2\theta_{1}}(M;x)}\left[\frac{1}{x}\right] + \right. \nonumber  \\
\left. \mathbb{E}_{F_{||\mathbf{h}||^2}(M,2,K;x)} \left[\frac{1}{x}\right] \mathbb{E}_{F_{\sin^2\theta_{2}}(M;x)}\left[\frac{1}{x}\right] + \mathbb{E}_{F_{||\mathbf{h}||^2}(M,1,K;x)} \left[\frac{1}{x}\right] \mathbb{E}_{F_{\sin^2\theta_{3}}(M;x)}\left[\frac{1}{x}\right]    \right).
\end{array}
\eeq
It's just a matter of algebra to compute these expectation similar to the $2$ user case as all the distributions have been given in appendix \ref{app:CDFs}.
\section{Semi-Orthogonal User Selection}
\label{app:SUS}
In SUS, the first user is selected with the largest channel norm but the second selected user is the one whose projection on the null space of the first user has the largest norm. Let's assume that user $1$ having channel $\mathbf{h_1}$ is the first selected user, hence the user with the largest norm, whose squared norm is distributed as $F_{||\mathbf{h}||^2}(M,1,K;x)$, the 1st order statistic among $K$ instances of $M$-dimensional channels. Let us further assume that user $2$ having channel $\mathbf{h_2}$ is the second selected user. This requires that $\mathbf{h_2} \sin(\theta_{1})$, the projection of channel vector $\mathbf{h_2}$ on the null space of the space spanned by $\mathbf{h_1}$, has the largest norm among $K-1$ users if these $K-1$ users' channel are projected on the null space of $\mathbf{h_1}$. Statistically this is the largest among $K-1$ norms in $M-1$ dimensional space (dimension reduction due to projection) conditioned upon the selection of the largest norm channel $\mathbf{h_1}$. Unfortunately this CDF is very hard to compute so we ease the computation using \cite[Lemma 3]{Maddah_ITIT08}, which was also used in \cite[Appendix III]{yoo_ZF_BF}. The term $||\mathbf{\check{h}_2}||^2 \stackrel{\Delta}{=} ||\mathbf{h_2||}^2 \sin^2(\theta_{1})$ is the maximum of $K-1$ channel norms orthogonalized w.r.t. $\mathbf{h_1}$. Following \cite{Maddah_ITIT08}, we can orthogonalize all the channel vectors w.r.t. an arbitrary vector so for each of them the squared norm is $\chi^2$ distributed with $2(M-1)$ degrees of freedom and each has the distribution which is given by $F_{||\mathbf{h}||^2}(M-1;x)$. Let us denote the projection of $\mathbf{h_i}$ on the null space of that arbitrary vector by $\mathbf{\tilde{h}_i}$, then the second largest norm of these orthogonalized vectors will be
\beq
||\mathbf{\hat{h}_2}||^2 = 2^{\mathrm{nd}} \max ||\mathbf{\tilde{h}_i}||^2, i=1,\ldots K
\eeq
whose distribution is given by $F_{||\mathbf{h}||^2}(M-1,2,K;x)$, the second largest of $K$ instances in $(M-1)$-dimensional space. Lemma $3$ in \cite{Maddah_ITIT08} shows that statistically $||\mathbf{\hat{h}_2}||^2$ is smaller than $||\mathbf{\check{h}_2}||^2$. The same procedure is repeated for the third iteration of the SUS and hence the third selected user is the third maximum of the $K$ users' channel norms which have been orthogonalized w.r.t. two arbitrary vectors. Hence the norm squared of the third user has the distribution of $F_{||\mathbf{h}||^2}(M-2,3,K;x)$, the third largest of the $K$ users in $(M-2)$-dimensional space. This procedure generalizes hence $i$-th selected user's squared channel norm would be distributed as $F_{||\mathbf{h}||^2}(M+1-i,i,K;x)$, $i$-th largest among $K$ users orthogonal to $(i-1)$-dimensional subspace. Thus the average of minimum transmit power required to satisfy SINR targets at $K_s$ users when they are selected through SUS is given by:
\beq
p_{\mathrm{S}}(K_s)  =  \sigma^2 \gamma \sum_{i=1}^{K_s} \left( \mathbb{E}_{F_{||\mathbf{h}||^2}(M+1-i,i,K;x)} \left[\frac{1}{x}\right] \right).
\eeq
We need to keep in mind that as the orthogonalized norms were replaced by their lower bounds in the derivation, the average transmit power from the above expressions $p_{\mathrm{S}}(K_s)$ will actually be the upper bound (performance lower bound) of the minimum power required with SUS.
\subsection{SUS for $2$ Users}
\label{app:SUS_2}
When only two users are selected through SUS, the average minimum power required can be computed by taking only the first two terms of the general SUS transmit power.
\beq
p_{\mathrm{S}} (2)  =  \sigma^2 \gamma \left( \mathbb{E}_{F_{||\mathbf{h}||^2}(M,1,K;x)} \left[\frac{1}{x}\right] + \mathbb{E}_{F_{||\mathbf{h}||^2}(M-1,2,K;x)} \left[\frac{1}{x}\right] \right).
\eeq
Computing the expectations gives the result of the corollary.
\subsection{SUS for $4$ Users}
\label{app:SUS_4}
When $4$ users are selected simultaneously and selection is done through SUS, the average minimum power required can be computed by taking the first four terms of the general SUS transmit power.
\beq
\begin{array}{l}
p_{\mathrm{S}} (4)  =  \sigma^2 \gamma \left( \mathbb{E}_{F_{||\mathbf{h}||^2}(M,1,K;x)} \left[\frac{1}{x}\right] + \mathbb{E}_{F_{||\mathbf{h}||^2}(M-1,2,K;x)} \left[\frac{1}{x}\right] + \mathbb{E}_{F_{||\mathbf{h}||^2}(M-2,3,K;x)} \left[\frac{1}{x}\right] + \right. \nonumber \\
\left. \mathbb{E}_{F_{||\mathbf{h}||^2}(M-3,4,K;x)} \left[\frac{1}{x}\right] \right).
\end{array}
\eeq
\section{Random User Selection}
\label{app:RUS}
In RUS, the users are selected randomly. Hence the norms and the directions of the channels of the selected users are randomly distributed. So the squared norms of the channel vectors for all selected users are distributed as $F_{||\mathbf{h}||^2}(M;x)$. Similar to the norm distributions, the directions of the selected users are also random and independent of each other. Hence $\sin^2\theta_i$ (where $\theta_i$ is the angle that a channel vector makes with an independent $i$-dimensional subspace is distributed as $F_{\sin^2\theta_i}(M;x)$. The angles (and $\sin^2$) distributions follow the same pattern as in NUS. So the average transmit power to reach SINR constraint of $\gamma$ at each selected user when these users are chosen using RUS is given by the following expression:
\beq
p_{\mathrm{R}}(K_s)  =  \sigma^2 \gamma \left( \mathbb{E}_{F_{||\mathbf{h}||^2}(M;x)} \left[\frac{1}{x}\right]\right) \sum_{i=1}^{K_s} \left( \mathbb{E}_{F_{\sin^2\theta_{(i-1)}}(M;x)} \left[\frac{1}{x}\right] \right).
\label{eq:rus_proof}
\eeq
The above expectations can be easily computed using the distributions given in appendix \ref{app:CDFs} and turn out to be:
\beq
\mathbb{E}_{F_{||\mathbf{h}||^2}(M;x)} \left[\frac{1}{x}\right] = \frac{1}{M-1}
\eeq
\beq
\mathbb{E}_{F_{\sin^2\theta_{(i-1)}}(M;x)} \left[\frac{1}{x}\right] = \frac{M-1}{M-i}
\eeq
Putting these values in equation (\ref{eq:rus_proof}), we get the expression for the average minimum transmit power to achieve SINR targets when users are selected through RUS:
\beq
p_{\mathrm{R}}(K_s)  =  \gamma \sigma^2 \sum_{i=1}^{K_s} \frac{1}{M-i}
\eeq
\section{Angle-Based User Selection}
\label{app:AUS}
When $2$ users are selected through AUS, the first selected user is the strongest user whose squared norm is distributed as $F_{||\mathbf{h}||^2}(M,1,K;x)$, the first order statistic of squared norm among $K$ users. As norms and directions are independent, the distribution of $\sin^2$ of the angle that other $K-1$ vectors individually make with the first selected vector (or $1$-dimensional subspace) all follow the distribution of $F_{\sin^2\theta_{1}}(M;x)$. The second selected user among $K-1$ users is the one making the largest angle with the first user. Hence statistically $\sin^2$ of this angle is the largest order statistic among $K-1$ instances and is distributed as $F_{\sin^2\theta_1}(M,1,K-1;x)$ (see appendix \ref{app:CDFs} for details). The squared norm of the second selected user is distributed as the squared norm of any random user which is not the user with the largest norm and hence the CDF is $F_{||\mathbf{h}||^2}(M,\acute{1},K;x)$, (see eq. (\ref{eq:CDF_not_largest}) in appendix \ref{app:CDFs}). We keep the same user ordering as detailed in NUS such that weaker user's signal gets decoded with less (no) interference. The average transmit power for this user selection is given by:
\beq
p_{\mathrm{A}} (2) =  \sigma^2 \gamma \left(\mathbb{E}_{F_{||\mathbf{h}||^2}(M,\acute{1},K;x)} \left[\frac{1}{x}\right] + \mathbb{E}_{F_{||\mathbf{h}||^2}(M,1,K;x)} \left[\frac{1}{x}\right] \mathbb{E}_{F_{\sin^2\theta_1}(M,1,K-1;x)} \left[\frac{1}{x}\right] \right).
\eeq
This will give the result for the case of two users. Unfortunately we could not extend the average power requirement with AUS to the general case of $K_s$ users due to added complexity.
\section{Performance Benchmark}
\label{app:LB}
To compute a lower bound on the minimum average transmit power (the performance upper bound) required to satisfy SINR targets of $\gamma$, we assume that the two selected users have the two largest norms as in NUS with CDFs as $F_{||\mathbf{h}||^2}(M,1,K;x)$ and $F_{||\mathbf{h}||^2}(M,2,K;x)$ and the angle between their channel vectors is the largest angle possible as in AUS, distributed as $F_{\sin^2\theta_1}(M,1,K-1;x)$. Hence with optimal ordering (the weaker user gets decoded with no interference), the lower bound on the average transmit power can be obtained by computing the expectations in the following expression:
\beq
p_{\mathrm{L}} (2)  =  \sigma^2 \gamma \left(\mathbb{E}_{F_{||\mathbf{h}||^2}(M,2,K;x)} \left[\frac{1}{x}\right] + \mathbb{E}_{F_{||\mathbf{h}||^2}(M,1,K;x)} \left[\frac{1}{x}\right] \mathbb{E}_{F_{\sin^2\theta_1}(M,1,K-1;x)} \left[\frac{1}{x}\right] \right).
\eeq
Like in AUS case, we could not extend this lower bound to the general case when $K_s$ users are selected due to the appearance of very complicated CDFs for norms and angles in the expression.
%
%
\bibliographystyle{IEEEbib}
\bibliography{strings,refs}

\begin{thebibliography}{10}

\bibitem{u_Asilomar09}
U.~Salim and D.~Slock,
\newblock ``Performance of different user selection algorithms for transmit
  power minimization,''
\newblock in {\em Proc. Asilomar Conference on Signals, Systems and Computers,
  Pacific Grove, CA, USA}, 2009.

\bibitem{cover_72}
T.~Cover,
\newblock ``Broadcast channels,''
\newblock {\em IEEE Trans. on Information Theory}, vol. 18, pp. 2--14, January
  1972.

\bibitem{telatar}
I.~E. Telatar,
\newblock ``Capacity of multi-antenna {G}aussian channels,''
\newblock {\em European Transactions on Telecommunications}, pp. 585--595,
  November 1999.

\bibitem{Foschini}
G.~J. Foschini and M.~J. Gans,
\newblock ``On limits of wireless communications in a fading environment when
  using multiple antennas,''
\newblock {\em Wireless Personal Communications}, vol. 6, pp. 311--335, 1998.

\bibitem{weingarten_bc}
H.~Weingarten, Y.~Steinberg, and S.~Shamai,
\newblock ``The capacity region of the {G}aussian multiple-input
  multiple-output broadcast channel,''
\newblock {\em IEEE Trans. on Information Theory}, vol. 52, pp. 3936--–3964,
  September 2006.

\bibitem{tse_bc}
P.~Viswanath and D.~Tse,
\newblock ``Sum capacity of the multiple antenna {G}aussian broadcast channel
  and uplink-downlink duality,''
\newblock {\em IEEE Trans. on Information Theory}, vol. 49, pp. 1912--1921,
  August 2003.

\bibitem{caire_bc}
G.~Caire and S.~Shamai (Shitz),
\newblock ``On the achievable throughput of a multiantenna {G}aussian broadcast
  channel,''
\newblock {\em IEEE Trans. on Information Theory}, vol. 49, pp. 1691–--1706,
  July 2003.

\bibitem{cioffi_bc}
W.~Yu and J.~M. Cioffi,
\newblock ``Sum capacity of {G}aussian vector broadcast channels,''
\newblock {\em IEEE Trans. on Information Theory}, vol. 50, pp. 1875--1892,
  September 2004.

\bibitem{costa_DPC}
M.~Costa,
\newblock ``Writing on dirty paper,''
\newblock {\em IEEE Trans. on Information Theory}, vol. 29, pp. 439 -- 441, May
  1983.

\bibitem{yoo_ZF_BF}
T.~Yoo and A.~Goldsmith,
\newblock ``On the optimality of multiantenna broadcast scheduling using
  zero-forcing beamforming,''
\newblock {\em IEEE Journal on Selected Areas in Communications}, vol. 24,
  March 2006.

\bibitem{yoo_ZF_Scheduling}
T.~Yoo, N.~Jindal, and A.~Goldsmith,
\newblock ``Multi-antenna downlink channels with limited feedback and user
  selection,''
\newblock {\em IEEE Journal on Selected Areas in Communications}, vol. 24,
  September 2007.

\bibitem{schubert_TVT04}
M.~Schubert and H.~Boche,
\newblock ``Solution of the multi-user downlink beamforming problem with
  individual {SINR} constraints,''
\newblock {\em IEEE Transactions on Vehicular Technology}, vol. 53, pp.
  18–--28, January 2004.

\bibitem{schubert}
M.~Schubert,
\newblock {\em Power-Aware Spatial Multiplexing with Unilateral Antenna
  Cooperation},
\newblock Ph.D. thesis, TU Berlin, 2003.

\bibitem{schubert_VTC02}
H.~Boche and M.~Schubert,
\newblock ``A general duality theory for uplink and downlink beamforming,''
\newblock in {\em Proc. IEEE Vehicular Technology Conference Fall}, September
  2002.

\bibitem{schubert_ISSSTA02}
M.~Schubert and H.~Boche,
\newblock ``Joint `dirty paper' pre-coding and downlink beamforming,''
\newblock in {\em Proc. ISSSTA}, September 2002, pp. 536–--540.

\bibitem{zhang_PIMRC07}
X.~Zhang, E.~Jorswieck, and B.~Ottersten,
\newblock ``User selection schemes in multiple antenna broadcast channels with
  guaranteed performance,''
\newblock in {\em Proc. IEEE Int. Symp. Personal, Indoor and Mobile Radio
  Communications}, 2007.

\bibitem{zhang_tx_power_eurasip}
X.~Zhang, E.~Jorswieck, B.~Ottersten, and A.~Paulraj,
\newblock ``{On the Asymptotic Optimality of Opportunistic Beamforming with
  Hard SINR Constraints},''
\newblock {\em EURASIP Journal on Advances in Signal Processing, special issue
  on Multiuser MIMO Transmission with Limited Feedback, Cooperation, and
  Coordination}, vol. 2009, pp. 1--12, 2009.

\bibitem{tse_book}
D.~Tse and P.~Viswanath,
\newblock {\em Fundamentals of {W}ireless {C}ommunications},
\newblock Cambridge, U.K. Cambridge Univ. Press, 2005.

\bibitem{cover_IT}
T.M. Cover and J.A. Thomas,
\newblock {\em Elements of {I}nformation {T}heory},
\newblock New York: John Wiley and Sons, 1991.

\bibitem{abramowitz}
M.~Abramowitz and A.~Stegun,
\newblock {\em Handbook of {M}athematical {F}unctions with {F}ormulas,
  {G}raphs, and {M}athematical {T}ables},
\newblock Dover Publications, 1964.

\bibitem{H_David}
H.~A. David,
\newblock {\em Order {S}tatistics},
\newblock New York: Wiley, 1980.

\bibitem{jindal_AntComb}
N.~Jindal,
\newblock ``Antenna combining for the {MIMO} downlink channel,''
\newblock {\em IEEE Trans. on Wireless Communications}, vol. 10, pp.
  3834–--3844, October 2008.

\bibitem{Gupta_beta}
A.~K. Gupta and S.~Nadarajah,
\newblock {\em Handbook of Beta Distribution and Its Applications},
\newblock CRC, 2004.

\bibitem{Maddah_ITIT08}
M.~A. Maddah-Ali, M.~A. Sadrabadi, and A.~K. Khandani,
\newblock ``Broadcast in {MIMO} systems based on a generalized {QR}
  decomposition: Signaling and performance analysis,''
\newblock {\em IEEE Trans. on Information Theory}, vol. 54, pp. 1124–--1138,
  March 2008.

\end{thebibliography}

\end{document}